\documentclass[11pt,a4paper]{article}
\usepackage[a4paper,margin=2.5cm]{geometry}

\usepackage[utf8]{inputenc}

\usepackage{microtype}
\usepackage[colorlinks,allcolors=blue]{hyperref}
\usepackage{mathtools}

\usepackage{amsfonts}
\usepackage{amsmath}
\usepackage{amssymb}
\usepackage{graphicx}
\usepackage{tikz}
\usepackage{nicefrac}

\usepackage{amsthm}
\newtheorem{theorem}{Theorem}
\newtheorem{lemma}[theorem]{Lemma}
\theoremstyle{definition}
\newtheorem{definition}[theorem]{Definition}
\theoremstyle{remark}
\newtheorem*{remark}{Remark}

\usepackage[ruled,vlined,linesnumbered]{algorithm2e}

\SetCommentSty{mycommfont}
\SetKw{Break}{break}

\newcommand{\Dhat}{\widehat{D}}

\title{Connectivity Oracles for Predictable Vertex Failures\thanks{Part of this work was done when Bingbing Hu and Adam Polak were at Max Planck Institute of Informatics, and Evangelos Kosinas was at University of Ioannina.}}
\author{Bingbing Hu \\ {\small UC San Diego} \and Evangelos Kosinas \\ {\small ISTA} \and Adam Polak \\ {\small Bocconi University}}
\date{}

\begin{document}

\maketitle

\begin{abstract}
The problem of designing connectivity oracles supporting vertex failures is one of the basic data structures problems for undirected graphs. It is already well understood: previous works [Duan--Pettie STOC'10; Long--Saranurak FOCS'22] achieve query time linear in the number of failed vertices, and it is conditionally optimal as long as we require preprocessing time polynomial in the size of the graph and update time polynomial in the number of failed vertices.

We revisit this problem in the paradigm of algorithms with predictions: we ask if the query time can be improved if the set of failed vertices can be predicted beforehand up to a small number of errors. More specifically, we design a data structure that, given a graph $G=(V,E)$ and a set of vertices predicted to fail $\widehat{D} \subseteq V$ of size $d=|\widehat{D}|$, preprocesses it in time $\tilde{O}(d|E|)$ and then can receive an update given as the symmetric difference between the predicted and the actual set of failed vertices $\widehat{D} \triangle D = (\widehat{D} \setminus D) \cup (D \setminus \widehat{D})$ of size $\eta = |\widehat{D} \triangle D|$, process it in time $\tilde{O}(\eta^4)$, and after that answer connectivity queries in $G \setminus D$ in time $O(\eta)$. Viewed from another perspective, our data structure provides an improvement over the state of the art for the \emph{fully dynamic subgraph connectivity problem} in the \emph{sensitivity setting} [Henzinger--Neumann ESA'16]. 

We argue that the preprocessing time and query time of our data structure are conditionally optimal under standard fine-grained complexity assumptions.
\end{abstract}

\section{Introduction}

Connectivity is a fundamental problem in undirected graphs. For a static input graph, one can compute its connected components in linear time; after such preprocessing it is easy to answer, in constant time, queries asking whether two given vertices $u$ and $v$ are connected.

The problem becomes much more challenging when the graph may change over time -- for instance, when certain vertices may fail and have to be removed from the graph. Duan and Pettie~\cite{DuanP10} were the first to propose a data structure that, after constructing it for graph $G=(V,E)$, can receive an \emph{update} consisting of a set of \emph{failed} vertices $D \subseteq V$, and then can answer connectivity \emph{queries} in $G \setminus D$ (i.e., in the subgraph of $G$ induced on $V \setminus D$). For a graph with $n$ vertices and $m$ edges, the data structure can be constructed in $\tilde{O}(nm)$ time\footnote{We use the $\tilde{O}$ notation to hide polylogarithmic $(\log n)^{O(1)}$ factors.}, the update takes time $\tilde{O}(d^6)$, where $d=|D|$, and each query runs in $O(d)$ time.\footnote{Their data structure allows also for a certain trade-off between the preprocessing and update times.} Under the Online Matrix-Vector Multiplication (OMv) Hypothesis, this linear query time is optimal (up to subpolynomial factors) if we require that preprocessing runs in time polynomial in the graph size and the update time is polynomial in the number of failed vertices~\cite{HenzingerKNS15}.

We study the problem of designing such connectivity oracles through the new lens of \emph{algorithms with predictions}. More specifically, we ask whether the query time can be improved, beyond the fine-grained lower bound, if some quite-accurate-but-not-perfect prediction of the set of failed vertices is available already during the preprocessing.

\subsection{Our results}

We define the problem of designing a connectivity oracle for predictable vertex failures as follows. We ask for a data structure that:
\begin{itemize}
    \item in the \emph{preprocessing phase}, receives an undirected graph $G=(V, E)$, an upper bound on the number of failed vertices $d \in \mathbb{Z}_+$, and a set $\Dhat \subseteq V$ of vertices predicted to fail ($|\Dhat| \leqslant d$);
    \item in the \emph{update phase}, receives a set $D \subset V$ of vertices that actually fail ($|D| \leqslant d$);
    \item in the \emph{query phase}, receives two vertices $u, v \in V \setminus D$, and answers whether $u$ and $v$ are connected in $G \setminus D \coloneqq G[V \setminus D]$.
\end{itemize}
We measure the prediction error as the size of the difference between the predicted and the actual set of failed vertices $\eta \coloneqq |\Dhat \triangle D| = |\Dhat \setminus D| + |D \setminus \Dhat|$. Our goal is to design a data structure that improves upon the running time of classic (prediction-less) data structures when $\eta \ll d$. To allow for update time sublinear in $d$, we let the set $D$ be given not explicitly but as the symmetric difference $\Dhat \triangle D \coloneqq (\Dhat \setminus D) \cup (D \setminus \Dhat)$ instead.

Our main result is such data structure with the following guarantees:
\begin{theorem}
\label{thm:algorithm}
There is a (deterministic) connectivity oracle for predictable vertex failures with $\tilde{O}(dm)$ preprocessing time and space, $\tilde{O}(\eta^4)$ update time, and $O(\eta)$ query time.
\end{theorem}

Our result can also be interpreted outside of the realm of algorithms with predictions. One can think of $\Dhat$ as of the set of vertices that are initially \emph{inactive}. Then, in the update phase, a small number $\eta$ of vertices change their state -- some inactive vertices become active and some active come to be inactive. With this perspective, our result can be interpreted as extending the type of the data structure of Duan and Pettie~\cite{DuanP10} with vertex \emph{insertions}. Since we want to keep the update and query times depending only on the number of affected vertices we cannot even afford to read all the edges of an inserted vertex, which can potentially have a high degree; hence, we need to know candidates for insertion already in the preprocessing phase, and the initially inactive vertices are precisely these candidates.

The above perspective -- of adding insertions to a Duan--Pettie-style oracle -- actually gives a good intuition of how our data structure works internally: we create a (prediction-less) vertex-failure connectivity oracle for $G \setminus \Dhat$, we use it to further remove $D \setminus \Dhat$, and we extend it so that it is capable of reinserting $\Dhat \setminus D$. We found that the original connectivity oracle of Duan and Pettie~\cite{DuanP10}, with its complex \emph{high-degree hierarchy tree}, is not best suited for such an extension. Instead, we base our data structure on a recent DFS-tree decomposition scheme proposed by Kosinas~\cite{Kosinas23}.

\paragraph{Comparison to prior work.}
Henzinger and Neumann~\cite{HenzingerN16} studied already the problem of designing a data structure like ours, under the name of \emph{fully dynamic subgraph connectivity} in the \emph{sensitivity setting}. Their proposed solution is a reduction to a Duan--Pettie-style (i.e., any deletion-only) connectivity oracle. Plugging in the current best oracle of Long and Saranurak~\cite{LongS22}, the reduction gives $\hat{O}(d^3m)$ preprocessing time\footnote{The $\hat{O}$ notation hides subpolynomial $n^{o(1)}$ factors.}, $\hat{O}(\eta^4)$ update time, and $O(\eta^2)$ query time. Because the oracle of~\cite{LongS22} is (near-)optimal under plausible fine-grained complexity assumptions, these running times turn out to be (nearly) the best we can hope to get from that reduction. 

Using algebraic techniques, van den Brand and Saranurak~\cite{BrandS19} tackle a more general problem of answering \emph{reachability} queries in \emph{directed} graphs in the presence of edge insertions and deletions. Note that in directed graphs edge updates can be used to simulate both vertex failures and activating (initially inactive) vertices. This is done by splitting each vertex into two -- one only for incoming edges, one only for outgoing -- and having a special edge from the former to the later whenever the original vertex is meant to be active. Therefore, the data structure that they provide also solves the problem that we study; it achieves $\hat{O}(n^\omega)$ preprocessing time, $\hat{O}(\eta^\omega)$ update time, and $O(\eta^2)$ query time, with high probability, where $\omega \leqslant 2.372$ denotes the matrix multiplication constant.

Our result offers a considerable improvement over the preprocessing and query times of both~\cite{HenzingerN16} and~\cite{BrandS19}.

\paragraph{Lower bounds.}
We argue that the preprocessing and update times of our data structure are optimal under standard fine-grained complexity assumptions. For the query time, notice that a classic (prediction-less) vertex-failure connectivity oracle can be simulated by setting $\Dhat = \emptyset$ in the preprocessing; then $\eta = d$, and hence the lower bound of Henzinger et al.~\cite{HenzingerKNS15} implies that the query time $O(\eta^{1-
\varepsilon})$ is impossible, for any $\varepsilon > 0$ (assuming the OMv Hypothesis and requiring that the preprocessing time is polynomial in the graph size and the update time is polynomial in $d$).

Regarding the preprocessing time, in Section~\ref{sec:lower_bound} we give a lower bound based on the Exact Triangle Hypothesis\footnote{The Exact Triangle Hypothesis says that there is $O(n^{3-\varepsilon})$ time algorithm for finding in an $n$-node edge-weighted graph a triangle whose edge weights sum up to $0$, for any $\varepsilon > 0$.}, which is implied by both the 3SUM Hypothesis and the APSP Hypothesis, so it is a weaker assumption than either of those popular fine-grained complexity hypotheses. Our lower bound shows that the $\tilde{O}(dm)$ complexity is conditionally optimal, up to subpolynomial factors, as long as the update and query times depend polynomially only on $\eta$, and not on the graph size nor on $d$.

\begin{theorem}
\label{thm:lowerbound}
Unless the Exact Triangle Hypothesis fails, there is no connectivity oracle for predictable vertex failures with preprocessing time $O(d^{1-\varepsilon}m)$ or $O(dm^{1-\varepsilon})$ and update and query times of the form $f(\eta) \cdot n^{o(1)}$, for any $\varepsilon > 0$.
\end{theorem}

We note that Long and Saranurak~\cite{LongS22} give a similar lower bound for connectivity oracles for vertex failures (without predictions), which is however based on the Boolean Matrix Multiplication Hypothesis, and hence it holds only against ``combinatorial'' algorithms.

\subsection{Related work}

\paragraph{Vertex-failure connectivity oracles.}

Duan and Pettie~\cite{DuanP10} were the first to study connectivity oracles supporting vertex failures, and they already achieved query time that turns out to be (conditionally) optimal~\cite{HenzingerKNS15}. Their preprocessing time and update time were subsequently improved~\cite{DuanP17,DuanP20,LongS22} to $\tilde{O}(dm)$ and $\tilde{O}(d^2)$, respectively, which are optimal under standard fine-grained complexity assumptions~\cite{LongS22}. Recently, Kosinas~\cite{Kosinas23} gave a data structure with a slightly worse $\tilde{O}(d^4)$ update time, but much simpler than all the previous constructions.

On the other hand, Pilipczuk et al.~\cite{PilipczukSSTV22} proposed a data structure that handles update and query together, in time doubly exponential in $d$ but with no dependence (even logarithmic) on the graph size. Their approach is thus beneficial, for instance, when $d$ is a constant.

\paragraph{Edge-failure connectivity oracles.}
Edge failures are much easier to handle than vertex failures. In particular, there exist edge-failure connectivity oracles using $\tilde{O}(m)$ space with $\tilde{O}(d)$ update time and $\tilde{O}(1)$ query time~\cite{PatrascuT07, KapronKM13, DuanP20}, which is \emph{unconditionally} optimal (up to polylogarithmic factors). This difference can be explained by the fact that a single edge deletion can only split the graph into two parts while deleting a high-degree vertex can create multiple new connected components.

\paragraph{Fully dynamic subgraph connectivity.}
In the fully dynamic subgraph connectivity problem every update consists in the activation or deactivation of a vertex, and there is no bound on the number of vertices that may be active or inactive at any given point. In the meantime, there are queries that ask whether two active vertices are connected through a path that visits only active vertices. This model was introduced by Frigioni and Italiano~\cite{FrigioniI00} in the context of planar graphs, where they show that one can achieve polylogarithmic time per operation. For general undirected graphs, there are several solutions, deterministic and randomized, that provide a trade-off between the update and query time \cite{ChanPR11,Duan10,DuanZ17,BaswanaCC019}. In all those cases, both the update and the query time are $\Omega(m^\delta)$, for some $\delta>0$, which is in accordance with known conditional lower bounds \cite{AbboudW14,HenzingerKNS15,JinX22}. Thus, there is a clear distinction between the complexity in this unrestricted model, against that in the fault-tolerant (sensitivity) setting. In other words, we can achieve better performance if there is a bound on the number of vertices at any given time whose state may be different from their original one.

\paragraph{Fault-tolerant labelling schemes.}
There is a recent line of work that provides efficient labelling schemes for connectivity queries in the fault-tolerant setting \cite{DoryP21,ParterP22a,IzumiEWM23,ParterPP24}. The problem here is to attach labels to the vertices, so that one can answer connectivity queries in the presence of failures given the labels of the query vertices and those of the failed vertices. The most significant complexity measures are the size of the labels (ideally, the number of bits per label should be polylogarithmic on the number of vertices), and the time to retrieve the answer from the labels. There are many problems that admit efficient labelling schemes (see, e.g., \cite{Peleg05}). Such labellings have applications in the distributed setting, because they allow to compute the answer without having to access a centralized data structure. All else being equal, it is obviously a much harder problem to provide labels whose total size matches that of the best centralized data structures.

\paragraph{Algorithms with predictions.}

The field of algorithms with predictions, also known as learning-augmented algorithms, blossoms in the recent years, see surveys~\cite{MitzenmacherV20,MitzenmacherV22} and an updated list of papers~\cite{LindermayrM22}. While the majority of those works concern online algorithms, data structures are also being studied in this paradigm, see, e.g.,~\cite{KraskaBCDP18,FerraginaV19,LinLW22}. These are however mostly index data structures, BSTs, etc., and not graph data structures.

Probably the closest to our work are three recent papers about dynamic graph algorithms with predictions~\cite{LiuS24,BFNP24,HSSY24}.  We note that \cite{BFNP24} also manage to use predictions to go beyond OMv-based conditional lower bounds, but they achieve it using algebraic algorithms for fast matrix multiplication, while we only use simple combinatorial methods and benefit from the fact that the predicted information is available already during the preprocessing time.

\paragraph{Recent follow-up work.}
Following the initial submission of our manuscript and posting a preprint on arXiv, Long and Wang~\cite{LongW24} improved over our update time, from $\tilde{O}(\eta^4)$ down to $\tilde{O}(\eta^2)$, which is conditionally optimal. Their data structure builds on the work of Long and Saranurak~\cite{LongS22}, and is arguably much more complex than ours. In particular, our bounds involve less logarithmic factors hidden in the $\tilde{O}$ notation.

\section{Preliminaries}

Let $G$ denote the input graph with $n$ vertices and $m$ edges. Let $\Dhat$ be the set of vertices predicted to fail, and let $D$ be the set of vertices that actually fail. Moreover, we let the size of both $D$ and $\Dhat$ to be at most $d$, and we let the prediction error $\eta \coloneqq |D\setminus \Dhat| + |\Dhat \setminus D|$. 

Without loss of generality, we may assume that $G\setminus\Dhat$ is connected. Indeed, we may add a new auxiliary vertex $z$ to $G$, that is connected with an edge with every vertex of the graph. Then, $G\setminus\Dhat$ is definitely connected, and in order to perform the updates, we always include $z$ in the set $D$ of failed vertices. Notice that this does not affect the asymptotic complexity of any of our measures of efficiency.

We build on the DFS-based tree decomposition scheme from \cite{Kosinas23}. Given the inputs $G$ and $\Dhat$, we grow an arbitrary DFS tree $T$ of $G \setminus \Dhat$ rooted at some vertex $r$. Let $T(v)$ denote the subtree of $T$ rooted at $v$. Let $\mathit{ND}(v)$ denote the number of descendants of $v$, which is equal to the size of $T(v)$. Let $p_T(v)$ denote the parent of $v$ in $T$.

Suppose we further remove some set of failed vertices $D \setminus \Dhat$ from $T$. Then $T$ gets decomposed into multiple connected components, which can be however connected to each other through \emph{back-edges} in $(G \setminus \Dhat) \setminus (D \setminus \Dhat)$. We use $r_C$ to denote the root of each connected component $C$. The presence of back-edges is important in that they establish connectivity between the otherwise loose connected components of $T \setminus (D \setminus \Dhat)$, so connectivity queries between two vertices can be reduced to queries between the connected component(s) that they reside in. Let $low(v)$ denote the lowest proper ancestor of $v$ that is connected to $T(v)$ through a back-edge. It is useful to extend the definition of $low(v)$ to $low_{k}(v)$, as is introduced in \cite[Section 2.1]{Kosinas23}. We identify $low_1(v)$ with $low(v)$. Then for every $k > 1$, we define $low_k(v)$ to be the next lowest ancestor of $v$ connected to $T(v)$ through a back-edge that is higher than $low_{k-1}(v)$. 

It is important to distinguish between two types of connected components: \emph{hanging subtree} and \emph{internal component}, which are introduced in \cite{Kosinas23}. A connected component $C$ of $T \setminus (D \setminus \Dhat)$ is called a \emph{hanging subtree} if none of the failed vertices from $D \setminus \Dhat$ is a descendant of it. Otherwise, $C$ is called an \emph{internal component}. (See \cite[Section 3.3]{Kosinas23} for more properties of these two types of components.) For a visual illustration of the decomposition, see the left hand side of Figure~\ref{fig:decomposition}. After the black vertices are removed, the white arrows represent internal components, and the grey triangles represent hanging subtrees. The reason for this distinction is that while the number of hanging subtrees can be close to $n$, the number of internal components is bounded by the number of failed vertices, which is a convenient fact that we leverage in the update phase. If $C$ is an internal component and $f \in D \setminus \Dhat$ is a failed vertex whose parent lies in $C$, then we call $f$ a \emph{boundary vertex of} $C$. The set of all boundary vertices of $C$ is denoted $\partial(C)$. For instance, in Figure~\ref{fig:decomposition}, $C_1$ has two boundary vertices and $C_3$ has only one boundary vertex.

We also add structure to the vertices in $D \setminus \Dhat$ by putting them in a \emph{failed vertex forest} $F$ (see \cite[Section 3.3]{Kosinas23}). Basically each failed vertex $f \in D \setminus \Dhat$ has a parent pointer $parent_F(f)$ to the largest (w.r.t.~the DFS numbering) failed vertex that is an ancestor of $f$. Each $f$ also has a pointer to the list of its children in the forest. This forest can be built in $O(\eta^2)$ time. To check whether a failed vertex $f$ is a boundary vertex of some internal component $C$, we can check whether $parent_F(f) \neq p_T(f)$~\cite[Lemma 5]{Kosinas23}.

\section{Our data structure}

\label{sec:algorithm}

Our oracle operates in three phases: preprocessing, update, and query.

\paragraph{Preprocessing.} We begin with a DFS tree $T$ for the graph $G \setminus \Dhat$, and prepare relevant data structures, which are detailed in Section~\ref{sec:preprocessing}.

\paragraph{Update.} In this phase, we are concerned with further removing the set of vertices $D \setminus \Dhat$ from $G \setminus \Dhat$, and adding back the vertices $\Dhat \setminus D$ that were deleted in the previous phase because of prediction errors. We build an auxiliary \emph{connectivity graph} $\mathcal{M}$, which captures connectivity relations among so called \emph{internal components} of $T \setminus (D \setminus \Dhat)$ and the reinserted vertices $\Dhat \setminus D$.

\paragraph{Query.} Given two vertices $s, t \in G \setminus D$, we find a representative connected component for each of $s$ and $t$, and answer the connectivity query based on $\mathcal{M}$.

\subsection{The preprocessing phase}
\label{sec:preprocessing}
In the preprocessing phase, we build a number of data structures around $G \setminus \Dhat$. It takes $\tilde{O}(dm)$ time. Our preprocessing data structures expand on the data structures introduced in \cite{Kosinas23} (the first six in the list below). See Section 3.1 in \cite{Kosinas23} for more explanation and ways to construct them.
\begin{enumerate}
    \item A DFS-tree $T$ for $G \setminus \widehat{D}$ rooted at some vertex $r$. For each vertex $v$ in $G \setminus \Dhat$, we compute its depth in $T$ and number of descendants $\mathit{ND}(v)$. We store the DFS numbering and identify the vertices with the order in which they are visited.
    \item A level-ancestor data structure on $T$~\cite{BenderF04}. For any vertex $v \in T$ and any depth $\ell \in \mathbb{Z}_+$, this data structure can return in constant time the ancestor of $v$ at depth $\ell$ in the tree.
    \item A 2D-range-emptiness data structure~\cite{ChanLP11,ChanT17} that can answer, in $\tilde{O}(1)$ time, whether there is a back-edge with one end in some segment\footnote{When we talk about \emph{segments} of a DFS tree $T$, we mean that we identify vertices w.r.t.~their DFS numbering, and a segment $[X,Y]$ of $T$ is the set of vertices $\{X,X+1,\ldots,Y\}$.} $[X_1, X_2]$ of $T$ and the other end in some segment $[Y_1, Y_2]$ of $T$.\label{item:2d}
    \item The $low_i$ points of all vertices in $T$, for $i \in \{1, \ldots, d\}$.
    \item For every $i \in \{1, \ldots, d\}$, we first sort the children of each vertex $v$ in $T$ in increasing order with respect to their $low_i$ points. Then we compute a new DFS numbering, which we denote $T_i$, corresponding to the DFS traversal in which children of every vertex are visited in that sorted order. Note that the ancestor-descendant relation in $T_i$ is the same as that in $T$.
    \item For every $T_i$, a separate 2D-range-emptiness data structure described in item~\ref{item:2d}.
    \item For every $u \in \Dhat$, we traverse the vertices in $T$ in a bottom-up fashion and mark vertices that have in their subtrees a neighbor of $u$. Then we create a tree $T_u$ as follows: for every vertex $v$ in $T$, rearrange the children of $v$ by putting unmarked vertices before marked vertices, and do a DFS traversal under this rearrangement to get $T_u$. Note that the ancestor--descendant relation in $T_u$ is the same as that in $T$.
    \item For every $u\in\Dhat$, a 2D-range-emptiness data structure on $T_u$ as described in item~\ref{item:2d}. \label{item:2dTu}
    \item For every $u, v \in \Dhat, u \neq v$, we sort the adjacent vertices of $v$ in $G \setminus \Dhat$ in increasing order with respect to their numbering in $T_{u}$ and store them as $neighbors_{u}(v)$. Moreover, in $neighbors(v)$ we store the adjacent vertices of $v$ in $G \setminus \Dhat$ in increasing order w.r.t. their original DFS numbering in $T$.   
\end{enumerate}

\begin{remark}
     The purpose of items 1--9 will become clear when we describe how to add edges to the auxiliary connectivity graph $\mathcal{M}$ in the update phase. Items 1--6 allow us to efficiently establish connectivity within $(G \setminus \Dhat) \setminus (D \setminus \Dhat)$. Items 7--9 allow us to efficiently establish connectivity between $\Dhat \setminus D$ and $(G \setminus \Dhat) \setminus (D \setminus \Dhat)$.
\end{remark}

The total time for initializing items 1--6 is $\tilde{O}(dm)$. See \cite[Section 3.1]{Kosinas23} for more explanation. Items 7, 8, and 9 take $\tilde{O}(dm)$ time because for every $u \in \Dhat$, we create a new tree by taking the old tree $T$ and rearranging its children lists in $\tilde{O}(m)$ time. Hence, the total preprocessing time is $\tilde{O}(dm)$, and the created data structures take $\tilde{O}(dm)$ space.

\subsection{The update phase}

\begin{figure}
\centering
\input{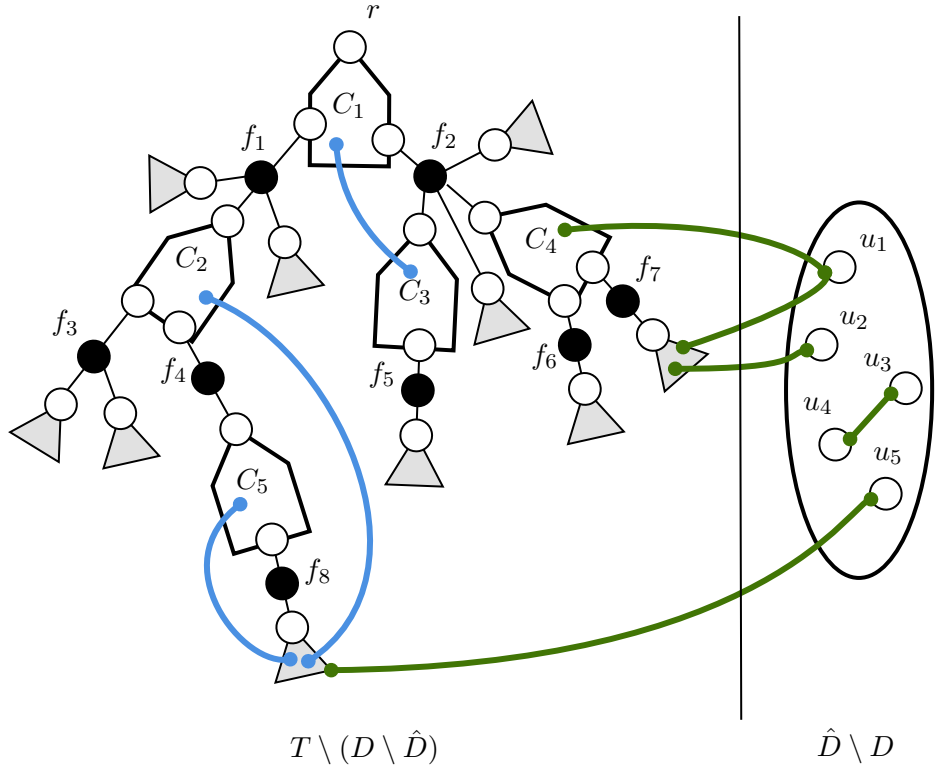}
\caption{The decomposition graph after we further remove $D \setminus \Dhat$ from $T$ and add back $\Dhat \setminus D$. On the left hand side we have $T \setminus (D \setminus \Dhat)$. The black vertices $f_1, \ldots, f_8$ are failed vertices in $D \setminus \Dhat$. The white pentagonal arrows represent internal components. The light grey triangles represent hanging subtrees. The blue edges represent back-edges between connected components of $T \setminus (D \setminus \Dhat)$. On the right hand side we have $\Dhat \setminus D$. The edges with an end in $\Dhat \setminus D$ are colored green.}
\label{fig:decomposition}
\end{figure}

\begin{figure}
\centering
\begin{tikzpicture}[x=0.75pt,y=0.75pt,line width=0.75pt]

\node [circle,draw,fill={rgb,255:red,176;green,228;blue,255}] (rC1) at (193,230) {$r_{C_1}$};
\node [circle,draw,fill={rgb,255:red,176;green,228;blue,255}] (rC3) at (224,180) {$r_{C_3}$};
\node [circle,draw,fill={rgb,255:red,255;green,246;blue,170}] (rC2) at (161,174) {$r_{C_2}$};
\node [circle,draw,fill={rgb,255:red,255;green,246;blue,170}] (rC5) at (184,113) {$r_{C_5}$};
\node [circle,draw,fill={rgb,255:red,191;green,176;blue,255}] (rC4) at (287,179) {$r_{C_4}$};
\node [circle,draw,fill={rgb,255:red,191;green,176;blue,255}] (u1)  at (394,219) {$u_1$};
\node [circle,draw,fill={rgb,255:red,191;green,176;blue,255}] (u2)  at (375,176) {$u_2$};
\node [circle,draw,fill={rgb,255:red,255;green,192;blue,184}] (u3)  at (418,161) {$u_3$};
\node [circle,draw,fill={rgb,255:red,255;green,192;blue,184}] (u4)  at (385,122) {$u_4$};
\node [circle,draw,fill={rgb,255:red,255;green,246;blue,170}] (u5)  at (420, 83) {$u_5$};

\draw (rC1) -- node [midway,xshift=-5,yshift=-2] {$1$} (rC3);
\draw (rC2) -- node [left                      ] {$2$} (rC5);
\draw (rC2) -- node [above                     ] {$6$} (u5);
\draw (rC5) -- node [below                     ] {$6$} (u5);
\draw (rC4) -- node [above                     ] {$4$} (u1);
\draw (u1)  -- node [midway,xshift= 5,yshift=-2] {$5$} (u2);
\draw (u3)  -- node [midway,xshift= 5,yshift=-3] {$3$} (u4);

\draw [dotted] (333,256) -- (333,62) ;

\end{tikzpicture}
\caption{The connectivity graph $\mathcal{M}$ for the example shown in Figure~\ref{fig:graph}. Edges are labeled with their types. The left hand side of the dotted line are vertices representing the internal components of $T \setminus (D \setminus \Dhat)$; the right hand side of the dotted line are vertices in $\Dhat \setminus D$. Edge $(r_{C_1}, r_{C_3})$ is due to a back-edge between $C_1$ and $C_3$. Edges $(r_{C_2}, r_{C_5}), (r_{C_2}, u_5), (r_{C_5}, u_5), (u_1, u_2)$ are all due to mutual connections to a hanging subtree. Edges $(r_{C_4}, u_1), (u_3, u_4)$ are both due to direct edges. Vertices in the same connected component are filled with the same color.}
\label{fig:graph}
\end{figure}
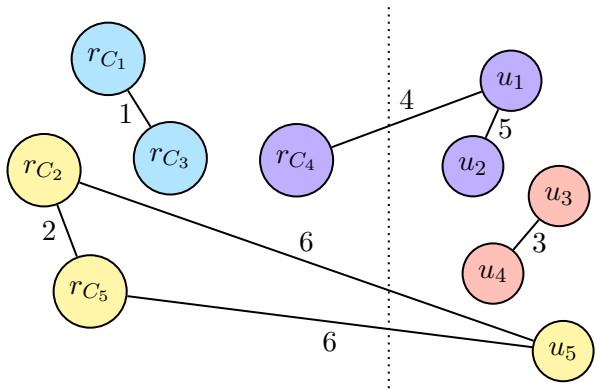

In the update phase, we receive the real set of failed vertices $D$. Since $\Dhat$ has already been removed from $G$ before computing our DFS tree $T$, we just need to further remove $D \setminus \Dhat$ from $T$, and then add back $\Dhat \setminus D$, which we needlessly removed (because of prediction errors) in the preprocessing phase. See Figure~\ref{fig:decomposition} for the decomposition of $G$ after these operations have been performed.

An issue with the full decomposition graph is that the number of hanging subtrees can be as large as order of $n$. We want to ``shrink'' the decomposition graph by getting rid of the hanging subtrees while still preserving connectivity. For this purpose, we define an auxiliary \emph{connectivity graph}, which builds on and extends a similar construction already introduced in \cite{Kosinas23}.

\begin{definition}
Let $\mathcal{M}$ be a graph with the vertex set
\[V(\mathcal{M}) \coloneqq (\Dhat \setminus D) \cup \{r_C : C \text{ is an internal component of } T \setminus (D \setminus \Dhat)\}.\] The edge set $E(\mathcal{M})$ consists of edges of the following six types (the first two types are the same as in \cite{Kosinas23}).
\begin{enumerate}
    \item If there is a back-edge connecting two internal components $C_1, C_2$ of $T \setminus (D \setminus \Dhat)$, then we add an edge between $r_{C_1}$ and $r_{C_2}$ in $\mathcal{M}$.
    \item If there is a hanging subtree $H$ of $T \setminus (D \setminus \Dhat)$ which is connected to internal components $C_1, \ldots, C_k$ through back-edges, with $C_k$ being an ancestor of all $C_1, \ldots, C_{k-1}$, then we add edges $(r_{C_k}, r_{C_1}), \ldots, (r_{C_k}, r_{C_{k-1}})$ in $\mathcal{M}$.
    \item If there is an edge connecting two vertices $u_1, u_2 \in \Dhat \setminus D$ in $G$, then we also add an edge between $u_1$ and $u_2$ in $\mathcal{M}$.
    \item If there is an edge connecting an internal component $C$ of $T \setminus (D \setminus \Dhat)$ and a vertex $u \in \Dhat \setminus D$, then we add an edge between $r_{C}$ and $u$ in $\mathcal{M}$.
    \item If there are two vertices $u_1, u_2 \in \Dhat \setminus D$ and a hanging subtree $H$ of $T \setminus (D \setminus \Dhat)$ such that both $u_1$ and $u_2$ are connected to $H$, then we add an edge between $u_1$ and $u_2$ in $\mathcal{M}$.
    \item If there are an internal component $C$ of $T \setminus (D \setminus \Dhat)$, a vertex $u \in \Dhat \setminus D$, and a hanging subtree $H$ of $T \setminus (D \setminus \Dhat)$ such that there is a back-edge connecting $C$ and $H$, and there is an edge connecting $u$ and $H$, then we add an edge between $r_{C}$ and $u$ in $\mathcal{M}$.
\end{enumerate}
As an illustration, in Figure~\ref{fig:graph} we show the connectivity graph for the graph in Figure~\ref{fig:decomposition}.
\end{definition}

The validity of type-2 edges comes from the following property of back-edges: If $e$ is a back-edge of $T \setminus (D \setminus \Dhat)$, then we have that either (i) the two ends of $e$ are in two internal components that are related as ancestor and descendant, or (ii) one end of $e$ is in a hanging subtree $H$ and the other end lies in an internal component that is an ancestor of $H$ (see Lemma 6 and Corollary 7 in \cite{Kosinas23} for a full proof). Hence, the internal components $C_1, \ldots, C_k$ in the characterization of type-2 edges are all ancestors of $H$, so it suffices to pick the most ancestral internal component $C_k$ and connect it by type-2 edges with the remaining components $C_1, \ldots, C_{k-1}$.

In the next lemma we show that $\mathcal{M}$ correctly captures the connectivity between internal components of ${T \setminus (D \setminus \Dhat)}$ and $\Dhat \setminus D$ in $G \setminus D$.

\begin{lemma}
Let $S$ and $S'$ each be an internal component of $T \setminus (D \setminus \Dhat)$ or a vertex in $\Dhat \setminus D$. Then $S$ and $S'$ are connected in $G \setminus D$ if and only if the vertices representing them in $\mathcal{M}$, $S_\mathcal{M}$ and $S'_{\mathcal{M}}$, are connected in $\mathcal{M}$.
\end{lemma}
\begin{proof}
First, let us show the ``only if'' direction. Let $S = P_0 - P_1 - P_2 - \dots - P_{k-1} - P_k = S'$ denote a path from $S$ to $S'$ in $G \setminus D$, where each $P_i$ is either an internal component of $T \setminus (D \setminus \Dhat)$, a hanging subtree of ${T \setminus (D \setminus \Dhat)}$, or a vertex in $\Dhat \setminus D$. Note that by the property of back-edges above, we cannot have two consecutive hanging subtrees in the path. 

Let us consider a segment $(P_i - P_{i+1})$ of this path. If neither $P_i$ nor $P_{i+1}$ is a hanging subtree, then we automatically have that their representatives in $\mathcal{M}$ are connected, by one of the constructions of type-1, type-3, or type-4 edges.

However, if one of $P_i$ and $P_{i+1}$ is a hanging subtree, and without loss of generality we assume it is $P_i$, then it follows that $P_{i-1}$ exists and it is not a hanging subtree. If $P_{i-1}$ and $P_{i+1}$ are both internal components, then their  representatives in $\mathcal{M}$ are (perhaps indirectly) connected by type-2 edges. If $P_{i-1}$ and $P_{i+1}$ are both vertices in $\Dhat \setminus D$, then their representatives are connected in $\mathcal{M}$ by a type-5 edge. If one of them is an internal component and the other is a vertex in $\Dhat \setminus D$, then their representatives are connected in $\mathcal{M}$ by a type-6 edge. Therefore, by applying the connectivity to all segments of the path, we get that the representatives of $S$ and $S'$ in $\mathcal{M}$ are connected, as desired.

Now, let us show the more straightforward ``if'' direction. Suppose $S_\mathcal{M}$ and $S'_{\mathcal{M}}$ are connected in $\mathcal{M}$. Then $S$ and $S'$ are either directly connected with an edge in $G \setminus D$, or indirectly connected through a hanging subtree in $G \setminus D$. In either case, we have that $S$ and $S'$ are connected.
\end{proof}

Note that only vertices from $\Dhat \setminus D$ and the roots of internal components of ${T \setminus (D \setminus \Dhat)}$ are present in $V(\mathcal{M})$. Hence, the number of vertices in $\mathcal{M}$ is $O(\eta)$. We argue that we can add in type-1 to type-6 edges in $\tilde{O}(\eta^4)$ time. Proposition~11 and Proposition~12 in \cite{Kosinas23} describe how we can compute type-1 edges in $\tilde{O}(\eta^2)$ time and type-2 edges in $\tilde{O}(\eta^4)$ time. Computing type-3 edges is also straightforward. Hence, we focus on computing type-4, type-5, and type-6 edges in the following sections. We finish the update phase by computing connected components of $\mathcal{M}$, in time $O(|V(\mathcal{M})| + |E(\mathcal{M})|) = O(\eta^2)$, so that in the query phase we are able to decide in constant time whether two vertices in $V(\mathcal{M})$ are connected.

\subsubsection{Computing type-4 edges}

In this section we describe our algorithm for adding type-4 edges to $\mathcal{M}$, see Algorithm~\ref{alg:type5}. Recall that type-4 edges connect vertices from $\Dhat \setminus D$ with internal components in which they have neighbors.

\begin{algorithm}
    \caption{For each internal component $C$ of $T \setminus (D \setminus \Dhat)$ and $u \in \Dhat \setminus D$, adds a type-4 edge between $r_C$ and $u$ in $\mathcal{M}$ if there is an edge in $G$ connecting $u$ to $C$.}
    \label{alg:type5}

    \ForEach{$u \in \Dhat \setminus D$} {
      \ForEach{{\rm internal component} $C$ of $T \setminus (D \setminus \Dhat)$} {
        $f_1, f_2, \ldots, f_k \gets$ boundary vertices of $C$ sorted in increasing order\;
        \For{$i = 0, \ldots, k$}{
          \leIf{$i>0$}{$L \gets f_i + \mathit{ND}(f_i)$}{$L \gets r_C$}
          \leIf{$i<k$}{$R \gets f_{i+1}-1$}{$R \gets r_C + \mathit{ND}(r_C)-1$}
          \tcp{Check the condition below in O(log n) time using binary search.}
          \If{$neighbors(u) \cap [L, R] \neq \emptyset$} {
            add an edge between $r_C$ and $u$ in $\mathcal{M}$\;
          }
        }
      }
    }
\end{algorithm}

For each vertex $u \in \Dhat \setminus D$ and each internal component $C$ of $T \setminus (D \setminus \Dhat)$ we want to efficiently establish if there is an edge between them. Let $f_1, f_2, \ldots, f_k$ be the boundary vertices of $C$. By Lemma~4(3) of \cite{Kosinas23}, $C$ can be written as the union of the following segments: $[r_C, f_1-1], [f_1+\mathit{ND}(f_1), f_2-1], \ldots, [f_{k-1}+\mathit{ND}(f_{k-1}), f_{k}-1], [f_k+\mathit{ND}(f_k), r_C+\mathit{ND}(r_C)-1]$. For each of these segments we check whether $neighbors(u)$ contains a vertex with the DFS number in that segment and if it does we add an edge between $u$ and $r_C$ in $\mathcal{M}$. Each such check can be performed efficiently, in time $O(\log(|neighbors(u)|)) = O(\log n)$, by using binary search. Indeed, for a segment $[L, R]$, we first binary search in $neighbors(u)$ the smallest number greater than or equal to $L$, and then just check if it is less than or equal to $R$.

For each internal component $C$, the number of binary searches we make is $O(|\partial( C)|)$. Since the total number of boundary vertices is $O(\eta)$, we have that the total number of binary searches in Algorithm~\ref{alg:type5} is $O(\eta^2)$, and hence its running time is $O(\eta^2 \log n)$.

\subsubsection{Computing type-5 edges}

In this section we describe our algorithm for adding type-5 edges to $\mathcal{M}$, see Algorithm~\ref{alg:type4}. Recall that type-5 edges join pairs of vertices in $\Dhat \setminus D$ that are connected in $G \setminus D$ via a hanging subtree.

For every $u, v \in \Dhat \setminus D$, we want to know whether they are connected to the same hanging subtree $H$. If we checked this separately for every tuple $(u, v, H)$, then it would be inefficient because the number of hanging subtrees can be as large as order of $n$. However, if we group together vertices that have edges to $u$ in their subtrees, which is exactly what we do in $T_u$, then we can process the hanging subtrees in batches. Let $f$ be a failed vertex in $D\setminus\Dhat$, and let $L_u(f)$ be the sequence of children of $f$ that are marked by $u$, sorted by their DFS numbering in $T_u$. Since we only care about hanging subtrees, we disregard vertices in $L_u(f)$ which are roots of internal components. This breaks $L_u(f)$ into $O(\eta)$ \emph{slices} consisting of roots of hanging subtrees. Each such slice $s = (L, \ldots, R) \in S_u(f)$ corresponds to (potentially) multiple hanging subtrees -- each subtree containing a neighbor of $u$ -- which form a contiguous set of vertices $[L, R + \mathit{ND}(R) - 1]$ in the DFS numbering for $T_u$. We can efficiently check if $v$ has a neighbor with the DFS number in that range. Indeed, we binary search in $neighbors_u(v)$ the smallest number greater than or equal to $L$, and check if it is less than or equal to $R + \mathit{ND}(R) - 1$. If this is the case, then $u$ and $v$ both have edges to a common hanging subtree, so we add a type-5 edge between them.

\begin{algorithm}
    \caption{Adds a type-5 edge in $\mathcal{M}$ between every $u, v \in \Dhat \setminus D$ that are connected to the same hanging subtree.}
    \label{alg:type4}

    \ForEach{$u \in \Dhat \setminus D$} {
    \tcp{Consider the DFS tree $T_u$ and the respective DFS numbering.}
        \ForEach{$v \in \Dhat \setminus D$} {
            \ForEach{{\rm failed vertex} $f \in D \setminus \Dhat$} {
                $S_u(f)$ $\gets$ the collection of contiguous slices of the sorted list of children of $f$ that are marked by $u$ consisting only of roots of hanging subtrees\;
                \ForEach{{\rm slice} $s \in S_u(f)$} {
                    $L \gets \min(s)$\;
                    $R \gets \max(s)$\;
                    \tcp{Check the condition below in O(log n) time using binary search.}
                    \If{$neighbors_u(v) \cap [L, R + \mathit{ND}(R) - 1] \neq \emptyset$\label{line:1drt4}} {
                        add an edge between $u$ and $v$ in $\mathcal{M}$\;
                    }
                }
            }
        }
    }
\end{algorithm}

We claim that the number of binary searches we make in line~\ref{line:1drt4} of Algorithm~\ref{alg:type4} is $O(\eta^3)$. Each of them corresponds to a tuple $(u, v, f, s)$. We leverage the fact that the number of internal components in $T \setminus (D \setminus \Dhat)$ is $O(\eta)$. Suppose $f_1, \ldots, f_k$ are all the failed vertices in $D \setminus \Dhat$. Then $|S_u(f_1)| + |S_u(f_2)| + \cdots + |S_u(f_k)| = O(\eta)$ for every $u \in \Dhat \setminus D$. It follows that Algorithm~\ref{alg:type4} runs in time $O(\eta^3 \log n)$.

\subsubsection{Computing type-6 edges}

In this section we describe our algorithm for adding type-6 edges to $\mathcal{M}$, see Algorithm~\ref{alg:type6}. Recall that type-6 edges join vertices from $\Dhat \setminus D$ with internal components that are connected to them indirectly via a hanging subtree.

\begin{algorithm}
    \caption{Adds a type-6 edge in $\mathcal{M}$ between every $u \in \Dhat \setminus D$ and every internal component $C$ of $T \setminus (D \setminus \Dhat)$ that are connected to the same hanging subtree $H$.}
    \label{alg:type6}

    \ForEach{$u \in \Dhat \setminus D$} {
    \tcp{Consider the DFS tree $T_u$ and the respective DFS numbering.}
        \ForEach{{\rm failed vertex} $f \in D \setminus \Dhat$} {
            $S_u(f)$ $\gets$ the collection of contiguous slices of the sorted list of children of $f$ that are marked by $u$ consisting only of roots of hanging subtrees\label{line:segst6}\;
            $f' \gets f$\;
            \While{$f' \neq null$\label{line:whilet6}} {
                \If{$p_T(f') \neq parent_F(f')$} { \tcp{Vertex $f'$ is a boundary vertex of some internal component.}
                    $C \gets$ the internal component of $T \setminus (D \setminus \Dhat)$ such that $f' \in \partial(C)$\;
                    \ForEach{{\rm slice} $s \in S_u(f)$} {
                        $L \gets \min(s)$\;
                        $R \gets \max(s)$\;
                        \If{$\textup{2D\_range\_query}_u([L, R+\mathit{ND}(R)-1] \times [r_C, p_T(f')]) \neq \emptyset$\label{line:2dt6}} {
                            add an edge between $u$ and $r_C$ in $\mathcal{M}$\;
                        }
                    }
                }
                $f' \gets parent_F(f')$\;
            }
        }
    }
\end{algorithm}

Conceptually, we want to query each $(C, u, H)$-tuple (where $C$ is an internal component, $u$ is a vertex from $\Dhat \setminus D$, and $H$ is a hanging subtree) whether $H$ contains a neighbor of $u$ and a back-edge to $C$. Similarly to how we compute type-5 edges, for each $u \in \Dhat \setminus D$ we group together relevant hanging subtrees in order to bound the number of queries we make on them (line~\ref{line:segst6}). Note that if $C$ is an internal component connected to $H$ through a back-edge, then $C$ is an ancestor of $H$ (see Lemma 6 and Corollary 7 in \cite{Kosinas23}). Therefore, once we have fixed a vertex $u \in \Dhat \setminus D$ and a failed vertex $f \in D \setminus \Dhat$, we only need to traverse $T_u$ up from $f$ and perform queries on each encountered internal component (the while-loop in line~\ref{line:whilet6}). During that traversal, whenever we encounter an internal component $C$, we want to check whether there exists a back-edge between $C$ and a hanging subtree $H$ whose root is a child of $f$ and which contains a neighbor of $u$. We do it in line~\ref{line:2dt6} by sending a query to the 2D-range-emptiness data structure (created in item~\ref{item:2dTu} of the preprocessing phase) indexed by the DFS numbering of tree $T_u$. The validity of line~\ref{line:2dt6} comes from two lemmas in \cite{Kosinas23}, which we reprint here:

\begin{lemma} [\cite{Kosinas23} Lemma 4(2)]
\label{lemma:prev3.2}
    Let $C$ be an internal component of $T \setminus X$, where $X$ is any set of failed vertices. For every vertex $v$ that is a descendant of $C$, there is a unique boundary vertex of $C$ that is an ancestor of $v$.
\end{lemma}

\begin{lemma} [\cite{Kosinas23} Lemma 8]
\label{lemma:prev3.6}
    Let $C, C'$ be two connected components of $T \setminus X$, where $X$ is any set of failed vertices, such that $C'$ is an internal component that is an ancestor of $C$. Let $b$ be the boundary vertex of $C'$ that is an ancestor of $C$. Then there is a back-edge from $C$ to $C'$ if and only if there is a back-edge from $C$ whose lower end lies in $[r_{C'}, p_T(b)]$.
\end{lemma}

Now let us show the correctness of Algorithm~\ref{alg:type6}. Suppose that a vertex $u$ from $\Dhat \setminus D$ and an internal component $C$ are both connected in $G \setminus D$ to a hanging subtree $H$, and denote by $f \in D \setminus \Dhat$ the parent of the root of $H$. We claim that Algorithm~\ref{alg:type6} eventually finds $C$ in $T_u$ by traversing up from $f$. Indeed, by Lemma~\ref{lemma:prev3.2}, let $f'$ be the boundary vertex of $C$ that is an ancestor of $f$. Then, by Lemma~\ref{lemma:prev3.6}, we know that the 2D range query in line~\ref{line:2dt6} returns true. Conversely, if the 2D range query in line~\ref{line:2dt6} returns true for some $C$ and $f'$, then we can also conclude that $u$ is connected with $C$ through the mediation of a hanging subtree whose root is a child of $f$.

By the same reasoning as for Algorithm~\ref{alg:type4}, each query corresponds to a tuple $(u, f, f', s)$. For fixed $u$ and $f'$ the total number of segments $s$ in $T_u \setminus (D \setminus \Dhat)$, over all choices of $f$, is $O(\eta)$. Hence the number of 2D range queries Algorithm~\ref{alg:type6} makes is $O(\eta^3)$. Since each such query takes $O(\log n)$ time~\cite{ChanLP11}, the total running time is $O(\eta^3 \log n)$.

\subsection{The query phase}
Finally, we describe our algorithm for the query phase, see Algorithm~\ref{alg:query}. Given $s$ and $t$ in $G \setminus D$, we first determine where they lie in. If they lie in some internal component(s) of $T \setminus (D \setminus \Dhat)$ or in the extra set of vertices $\Dhat \setminus D$, we can check whether these components/vertices are connected in $\mathcal{M}$. However, if $s$ (or~$t$) lies in a hanging subtree, we try to find a surrogate for $s$ (respectively~$t$), namely a vertex in $\Dhat \setminus D$ or an internal component of $T \setminus (D \setminus \Dhat)$ that is connected by an edge to the hanging subtree of $s$ (respectively~$t$).\footnote{Note that in order to find a surrogate internal component it is sufficient to check the $low_1, \ldots, low_\eta$ points of the root of the hanging subtree.} If we can find such a surrogate, we can replace $s$ (respectively $t$) with it, and refer to $\mathcal{M}$ as in the previous case. However, if no such surrogate exists, either for $s$ or for $t$, then $s$ and $t$ are connected if and only if they reside in the same hanging subtree.

\begin{algorithm}
    \caption{Given two vertices $s, t \in V \setminus D$, answers if they are connected in~$G \setminus D$.}
    \label{alg:query}
    \SetKwProg{Fn}{function}{}{}
    \Fn{\textup{find\_representative($u$)}}{
      \lIf {$u \in \Dhat \setminus D$} {
        \Return $u$%
      }
      \lIf {$u$ lies in an internal component $C$ of $T \setminus (D \setminus \Dhat)$} {
        \Return $r_C$%
      }
      $H \gets $ the hanging subtree of $T \setminus (D \setminus \Dhat)$ where $u$ lies in\;
      \tcp{Try to find an internal component as a surrogate.}
      \For{$i \in \{1, \ldots, \eta\}$ \label{line:sur1query}} {
        \If{$low_i(r_H) \neq null$ \KwSty{and} $low_i(r_H) \notin D$} {
          $C \gets$ the internal component of $T \setminus (D \setminus \Dhat)$ where $low_i(r_H)$ lies in\;
          \Return $r_C$\;
        }
      }
      \tcp{Try to find a vertex in $\Dhat \setminus D$ as a surrogate.}
      \ForEach{$v \in \Dhat \setminus D$ \label{line:sur2query}} {
        \lIf{$r_H$ is marked in $T_v$} {
          \Return $v$%
        }
      }
      \tcp{No surrogate found, return the hanging tree's root as the representative.}
      \Return $r_H$\;
    }
    $s \gets $ find\_representative($s$)\;
    $t \gets $ find\_representative($t$)\;
    \eIf {$s \in V(\mathcal{M})$ \KwSty{and} $t \in V(\mathcal{M})$} {
      \Return whether $s$ and $t$ are connected in $\mathcal{M}$\;
    }(\tcc*[h]{At least one vertex lies in an isolated hanging subtree,  represented by its root.}){
      \Return whether $s = t$\;
    }
\end{algorithm}

Now let us justify the running time for Algorithm~\ref{alg:query}. Given that we have computed the connected components of $\mathcal{M}$ at the end of the update phase, we can decide in $O(1)$ time the connectivity among internal components and vertices from $\Dhat \setminus D$. Locating the correct connected component in $T \setminus (D \setminus \Dhat)$ and deciding whether it is an internal component or a hanging subtree takes $O(\eta)$ time (see Proposition~13 in \cite{Kosinas23} for more detail). The for-loops on lines \ref{line:sur1query} and \ref{line:sur2query}, which are looking for a surrogate, both take $O(\eta)$ time. Summing up, Algorithm~\ref{alg:query} answers a query in $O(\eta)$ time.

\section{Lower bound for preprocessing time}\label{sec:lower_bound}
In this section, we show that the $\tilde{O}(dm)$ preprocessing time of our oracle cannot be improved by a polynomial factor assuming the Exact Triangle Hypothesis, which is implied by both the 3SUM Hypothesis and the APSP Hypothesis.

To show this, we give a fine-grained reduction from the offline SetDisjointness problem. This problem was introduced by Kopelowitz, Pettie, and Porat~\cite{KopelowitzPP16}, who showed it is hard under the 3SUM Hypothesis. Vassilevska Williams and Xu extended their hardness result to hold under a weaker assumption, namely the Exact Triangle Hypothesis; as a consequence the problem is also hard under the APSP Hypothesis. The input to the SetDisjointness problem is a universe of elements $U$, a family of subsets $F \subseteq 2^U$, and a collection of query pairs $(S_i, S_j) \in F \times F$, whose number is denoted by $q$. For each query, one has to answer whether $S_i \cap S_j$ is empty or not. We use the following lower bound for SetDisjointness from \cite[Corollary~3.11]{WilliamsX20}.

\begin{theorem}[Vassilevska Williams and Xu~{\cite[Corollary~3.11]{WilliamsX20}}]
\label{thm:setdisjointness}
For any constant $\gamma \in (0, 1)$, let $\mathcal{A}$ be an algorithm for offline SetDisjointness where $|U| = \Theta(n^{2-2\gamma})$, $|F| = \Theta(n)$, each set in $F$ has at most $O(n^{1-\gamma})$ elements from $U$, and $q = \Theta(n^{1+\gamma})$. Assuming the Exact Triangle Hypothesis, $\mathcal{A}$ cannot run in $O(n^{2-\varepsilon})$ time, for any $\varepsilon > 0$.
\end{theorem}

\begin{proof}[Proof of Theorem~\ref{thm:lowerbound}]
We reduce solving the offline SetDisjointness problem, in the parameter regime specified in Theorem~\ref{thm:setdisjointness} (later, we set $\gamma = 1 - \varepsilon / 2$), to creating a connectivity oracle for predictable vertex failures on an undirected graph with $O(n + n^{2-2\gamma})$ vertices and $O(n^{2-\gamma})$ edges, with $d=O(n)$, and then running the update and query phases $q$ times each, with $\eta=O(1)$.

Here is how we construct the undirected graph $G$. Let $V(G)$ consist of two disjoint sets $A$ and $B$. The vertices in $A = \{a_1, a_2, \ldots, a_{|F|}\}$ represent the sets in $F = \{S_1, S_2, \ldots, S_{|F|}\}$, and the vertices in $B = \{b_u \mid u \in U\}$ represent the elements in $U$. We put edges between $A$ and $B$ so that $a_i \in A$ is connected to $b_u \in B$ if and only if $u \in S_i$. Note that, for $i \neq j$, $S_i \cap S_j \neq \emptyset$ if and only if $a_i$ and $a_j$ are connected in the subgraph of $G$ induced on $B \cup \{a_i, a_j\}$. The constructed graph has $|A| + |B| = |F| + |U| =  O(n + n^{2-2\gamma})$ vertices and $O(|F| \cdot \max_i |S_i|) = O(n^{2-\gamma})$ edges. 

In the preprocessing phase we set $\Dhat = A$. Then, for each input query pair $(S_i, S_j)$, we set $D = \Dhat \setminus \{a_i, a_j\}$, do the update, and then query connectivity between $a_i$ and $a_j$.

Let $t_1$ denote the preprocessing time; $t_2$ denote the update time; and $t_3$ denote the query time. Then, the total time to solve SetDisjointness is $t_1 + q \cdot (t_2 + t_3)$. By Theorem~\ref{thm:setdisjointness}, we get that $t_1 + q \cdot (t_2 + t_3) \geqslant n^{2-o(1)}$, unless the Exact Triangle Hypothesis fails. Since $\eta = 2$ is a constant, we know that $t_2 = n^{o(1)}$ and $t_3 = n^{o(1)}$, and hence $q \cdot (t_2 + t_3) = n^{1+\gamma+o(1)}$. This forces $t_1 \geqslant n^{2-o(1)}$. Recall that $d = |\Dhat| = |A| = |F| = O(n)$, and $m = O(n^{2-\gamma})$. If $t_1 = O(d^{1-\varepsilon}m)$ (or $O(dm^{1-\varepsilon})$), then by setting $\gamma = 1 - \varepsilon/2$ we would get that $t_1 = O(n^{2-\varepsilon/2})$ and hence the Exact Triangle Hypothesis would fail.
\end{proof}

\bibliographystyle{alphaurl}
\bibliography{main}

\end{document}